\documentclass[review]{elsarticle}

\usepackage{lineno,hyperref}
\usepackage{graphicx}
\usepackage{bm}
\usepackage{amsmath,amssymb,amsfonts}
\usepackage{amsthm}
\usepackage{derivative}
\usepackage{color}
\usepackage{balance}
\usepackage{enumitem}

\usepackage[margin=1.0in]{geometry}

\usepackage{stfloats}

\newtheorem{remark}{Remark}
\newtheorem*{theorem*}{Theorem}

\newtheorem*{definition*}{Definition}

\usepackage{fancyhdr}

\modulolinenumbers[5]

\journal{Journal of \LaTeX\ Templates}

\begin{document}

\begin{frontmatter}

\title{Convergence Analysis of Consensus-ADMM for General QCQP}
\tnotetext[mytitlenote]{The work of H. Huang is supported by the Graduate School CE within the Centre for Computational Engineering at Technische Universität Darmstadt.}


\author[mymainaddress]{Huiping Huang} 
\ead{h.huang@spg.tu-darmstadt.de}

\author[mysecondaryaddress]{Hing Cheung So}
\ead{hcso@ee.cityu.edu.hk}

\author[mymainaddress]{Abdelhak M. Zoubir}
\ead{zoubir@spg.tu-darmstadt.de}

\address[mymainaddress]{Department of Electrical Engineering and Information Technology, Technische Universität Darmstadt, Darmstadt, Germany}
\address[mysecondaryaddress]{Department of Electrical Engineering, City University of Hong Kong, Hong Kong, China}

\begin{abstract}
We analyze the convergence properties of the consensus-alternating direction method of multipliers (ADMM) for solving general quadratically constrained quadratic programs. We prove that the augmented Lagrangian function value is monotonically non-increasing as long as the augmented Lagrangian parameter is chosen to be sufficiently large. Simulation results show that the augmented Lagrangian function is bounded from below when the matrix in the quadratic term of the objective function is positive definite. In such a case, the consensus-ADMM is convergent.
\end{abstract}

\begin{keyword}
Augmented Lagrangian function, Consensus-ADMM, quadratically constrained quadratic program
\end{keyword}

\end{frontmatter}

\linenumbers

\section{Introduction}
\label{introduction_section}
An attractive algorithm, called consensus-alternating direction method of multipliers (ADMM), has been recently proposed in \cite{Huang2016}, for solving general quadratically constrained quadratic programs (QCQPs). Note that these problems could be non-convex and thus NP-hard, including several special cases in real-world applications, see e.g. \cite{Mehanna2015, Park2017, Mehanna2013, Liu2018, Hamza2019, Zheng2020, Cheng2021, Dan2022, Huang2022}. The consensus-ADMM, primarily enjoys the following two advantages: i) it formulates the general QCQP in such a way that each subproblem is a QCQP with only one constraint, referred to as QCQP-1, which can be effectively solved irrespective of (non)convexity; and ii) memory efficient implementation and parallel/distributed processing are possible, which greatly save the effort of solving general QCQPs, especially when the dimensions of the resulting matrices and/or the number of constraints are large.

Despite the effectiveness of consensus-ADMM, only a weak convergence result was presented in \cite{Huang2016}. That is, the authors claim that the consensus-ADMM converges to a Karush-Kuhn-Tucker (KKT) stationary point, under the key assumptions that the limit points do exist, and that the difference between two successive iterations converges to zero, see Theorem 1 in \cite{Huang2016}. Note that such assumptions are nonstandard and overly restrictive \cite{Hong2016}. In this paper, we analyze the convergence of the consensus-ADMM, and show that the augmented Lagrangian function is monotonically non-increasing as long as the Lagrangian parameter is larger than a certain value.

The remainder of the paper is organized as follows. The problem formulation and the consensus-ADMM are presented in Section \ref{problemstatement_section}. Convergence analyses are established in Section \ref{convergenceproperty_section}. Simulation results are shown in Section \ref{simulation_section}, while Section \ref{conclusion_section} concludes the paper.

\textit{Notation:} Bold-faced lower-case and upper-case letters denote vectors and matrices, respectively. ${\bf I}_{n}$ denotes the $n \times n$ identity matrix, and ${\bf 1}$ and ${\bf 0}$ are the all-one and all-zero vectors of appropriate size, respectively. Superscripts $\cdot^{\mathrm{T}}$, $\cdot^{\mathrm{H}}$ and $\cdot^{-1}$ stand for transpose, Hermitian transpose and inverse operators, respectively. $|\cdot|$ is the absolute operator. $\|\cdot\|_{2}$ denotes the $\ell_{2}$-norm of a vector. $\lambda_{\mathrm{max}}(\cdot)$ and $\lambda_{\mathrm{min}}(\cdot)$ denote the maximum and minimum eigenvalues, respectively. $\mathbb{C}$ is the set of complex numbers, and $\Re\{\cdot\}$ returns the real part of its input complex-valued variable. $\langle {\bf a} , {\bf b}\rangle = {\bf a}^{\mathrm{H}}{\bf b}$ is the inner product of vectors ${\bf a}$ and ${\bf b}$ of same sizes. $\max\{a,b\}$ returns the maximum between $a$ and $b$. ${\bf X} \succ 0$ means that ${\bf X}$ is positive definite. $\mathrm{E}\left\{ \cdot \right\}$ and $\mathrm{Var}\{ \cdot \}$ denote the expected value and variance of a random variable, respectively.

\section{Problem Statement and Consensus-ADMM}
\label{problemstatement_section}
General QCQP is an optimization problem that minimizes a quadratic function subject to quadratic inequality and equality constraints \cite{Luo2010}, which can be typically formulated as \cite{Huang2016}
\begin{subequations}
\label{general_QCQP_prob}
\begin{align}
\min_{{\bf x} \in \mathbb{C}^{n}} ~ & {\bf x}^{\mathrm{H}}{\bf A}_{0}{\bf x} - 2\Re\{{\bf b}_{0}^{\mathrm{H}}{\bf x}\} \\
\label{2_general_QCQP_prob}
\mathrm{s.t.} ~~~\! & {\bf x}^{\mathrm{H}}{\bf A}_{i}{\bf x} - 2\Re\{{\bf b}_{i}^{\mathrm{H}}{\bf x}\} \leq c_{i}, ~ \forall i = 1, 2, \cdots, m,
\end{align}
\end{subequations}
where the known matrices ${\bf A}_{i} \in \mathbb{C}^{n \times n}$ are assumed to be general Hermitian matrices (possibly indefinite), and ${\bf b}_{i} \in \mathbb{C}^{n}$, for all $i = 0, 1, \cdots, m$. Note that any quadratic equality constraint can always be reformulated as two inequality constraints of the form in (\ref{2_general_QCQP_prob}). Therefore, for simplicity, only inequality constraints are given in Problem (\ref{general_QCQP_prob}). The following standard steps \cite{Boyd2011} guide us to the consensus-ADMM for solving the general QCQP, i.e., Problem (\ref{general_QCQP_prob}).

Step i): We reformulate Problem (\ref{general_QCQP_prob}) by introducing auxiliary variables $\{{\bf z}_{i}\}_{i=1}^{m}$, and settling the original variable ${\bf x}$ and the auxiliary variables $\{{\bf z}_{i}\}_{i=1}^{m}$ in a separable manner, as
\begin{subequations}
\label{auxiliary_QCQP_prob}
\begin{align}
\min_{{\bf x}, \{{\bf z}_{i}\!\} } ~ & {\bf x}^{\mathrm{H}}{\bf A}_{0}{\bf x} - 2\Re\{{\bf b}_{0}^{\mathrm{H}}{\bf x}\} \\
\mathrm{s.t.} ~~~\! & {\bf z}_{i}^{\mathrm{H}}{\bf A}_{i}{\bf z}_{i} - 2\Re\{{\bf b}_{i}^{\mathrm{H}}{\bf z}_{i}\} \leq c_{i}, \\
& {\bf z}_{i} = {\bf x}, ~~~ \forall i = 1, 2, \cdots, m.
\end{align}
\end{subequations}

Step ii): We form the scaled-form augmented Lagrangian function according to Problem (\ref{auxiliary_QCQP_prob}), by dealing with the equality constraints therein, i.e., ${\bf z}_{i} = {\bf x}$, as
\begin{align}
\label{augLagFun}
\mathcal{L}({\bf x}, \{{\bf z}_{i}\}, \{{\bf u}_{i}\})  \triangleq  {\bf x}^{\mathrm{H}}{\bf A}_{0}{\bf x} - 2\Re\{{\bf b}_{0}^{\mathrm{H}}{\bf x}\} + \rho \! \sum_{i = 1}^{m} \left( \|{\bf z}_{i} - {\bf x} + {\bf u}_{i}\|_{2}^{2} - \|{\bf u}_{i}\|_{2}^{2} \right)
\end{align}
where $\rho > 0$ is the augmented Lagrangian parameter, and ${\bf u}_{i}$ is the scaled dual variable corresponding to the equality constraint ${\bf z}_{i} = {\bf x}$ in Problem (\ref{auxiliary_QCQP_prob}).

Step iii): The consensus-ADMM updating equations can be written down by separately solving $\{{\bf z}_{i}\}$ and ${\bf x}$, as
\begin{subequations}
\label{ADMM_k1_prob}
\begin{align}
\label{ADMM_z_k1}
{\bf z}_{i}^{(\! k+1 \!)}  & \! =  \left\{ \!\! \begin{array}{l}
\displaystyle\arg\min_{{\bf z}_{i}} ~ \|{\bf z}_{i} - {\bf x}^{(\! k \!)} + {\bf u}_{i}^{(\! k \!)}\|_{2}^{2} \\ 
~~~ \mathrm{s.t.} ~~~~~ {\bf z}_{i}^{\mathrm{H}}{\bf A}_{i}{\bf z}_{i} - 2\Re\{{\bf b}_{i}^{\mathrm{H}}{\bf z}_{i} \} \leq c_{i},
\end{array}
\right. \\
\label{ADMM_x_k1}
{\bf x}^{(\! k+1 \!)}   & \! =  ({\bf A}_{0} + m\rho {\bf I}_{n})^{-1} \! \left[ {\bf b}_{0} + \rho \sum_{i = 1}^{m}({\bf z}_{i}^{(\! k+1 \!)} \! + \! {\bf u}_{i}^{(\! k \!)}) \right],   \\
\label{ADMM_u_k1}
{\bf u}_{i}^{(\! k+1 \!)}  & \! =  {\bf u}_{i}^{(\! k \!)} + {\bf z}_{i}^{(\! k+1 \!)} - {\bf x}^{(\! k+1 \!)},
\end{align}
\end{subequations}
where the superscript $\cdot^{(\! k \!)}$ denotes the corresponding variable obtained at the $k$-th iteration. As has been mentioned in \cite{Huang2016}, we put the consensus-ADMM into this form due to the fact that each update of ${\bf z}_{i}$ is a QCQP-1, which can be computed optimally despite the fact that the matrices ${\bf A}_{i}$ may be indefinite. The readers are referred to \cite{Huang2016} for details of solving QCQP-1. Furthermore, in the form of (\ref{ADMM_k1_prob}), $\{ {\bf z}_{i} \}_{i = 1}^{m}$ can be updated in parallel, and the result of $({\bf A}_{0} + m\rho {\bf I}_{n})^{-1}$ can be cached to save computations in the subsequent iterations. On the other hand, it is worth pointing out that, different from \cite{Huang2016} in which the algorithm was presented by updating ${\bf x}$ before $\{{\bf z}_{i}\}$, we update $\{{\bf z}_{i}\}$ before ${\bf x}$. This shall be shown to be very important when analyzing the convergence of the algorithm.

\section{Convergence Analysis}
\label{convergenceproperty_section}
This section is devoted to a theorem showing the monotonicity of the augmented Lagrangian function. First of all, from (\ref{ADMM_u_k1}), we have $ {{\bf z}}_{i}^{(\! k+1 \!)} + {{\bf u}}_{i}^{(\! k \!)} = {{\bf u}}_{i}^{(\! k+1 \!)} + {\bf x}^{(\! k+1 \!)}$. Substituting it into (\ref{ADMM_x_k1}) yields
\begin{align}
\label{u_and_x}
\rho \sum_{i = 1}^{m}{{\bf u}}_{i}^{(\! k+1 \!)} = {\bf A}_{0}{\bf x}^{(\! k+1 \!)} - {\bf b}_{0}.
\end{align}

\begin{remark}
Equality (\ref{u_and_x}) is one of the keys to the success of proving the monotonicity of the augmented Lagrangian function, which will be shown subsequently. Moreover, such an equality or other similar results are not available if we update ${\bf x}$ before $\{ {\bf z}_{i} \}$, as done in \cite{Huang2016}.
\end{remark}

Now, we show the monotonicity of the augmented Lagrangian function, $\mathcal{L}({\bf x}, \{{\bf z}_{i}\}, \{{\bf u}_{i}\})$, defined in (\ref{augLagFun}).

\begin{theorem*}
\label{monotonicity_lemma}
As long as the augmented Lagrangian parameter
\begin{align}
\label{rho_condition}
\rho > \max\left\{ - \frac{\lambda_{\mathrm{min}}}{m} , \frac{\sqrt{mC}\lambda_{\mathrm{max}} + \max\{-\lambda_{\mathrm{min}} , 0\}}{m} \right\},
\end{align}
where $C$ stands for some constant, and $\lambda_{\mathrm{min}}$ and $\lambda_{\mathrm{max}}$ are short for $\lambda_{\mathrm{min}}({\bf A}_{0})$ and $\lambda_{\mathrm{max}}({\bf A}_{0})$, respectively, we have:
\begin{itemize}
	\item $\mathcal{L}({\bf x}^{(\! k+1 \!)} \!, \{{\bf z}_{i}^{(\! k+1 \!)}\} \!, \{{\bf u}_{i}^{(\! k+1 \!)}\}) \leq \mathcal{L}({\bf x}^{(\! k \!)} \!, \{{\bf z}_{i}^{(\! k \!)}\} \!, \{{\bf u}_{i}^{(\! k \!)}\})$, $\forall k = 0, 1, 2, \cdots$;
	\item $\displaystyle \lim_{k \to \infty}\left\|{\bf x}^{(\!k + 1 \!)} - {\bf x}^{(\!k\!)}\right\|_{2} = 0$.
\end{itemize}
\end{theorem*}

\begin{proof}
To show $\mathcal{L}( {\bf x}^{(\! k\!+\!1 \!)}, \! \{ \! {\bf z}_{i}^{(\! k\!+\!1 \!)} \! \}, \! \{ \! {\bf u}_{i}^{(\! k\!+\!1 \!)} \! \} ) \leq \mathcal{L}( {\bf x}^{(\! k \!)}, \! \{ \! {\bf z}_{i}^{(\! k \!)} \! \}, \! \{ \! {\bf u}_{i}^{(\! k \!)} \! \} )$ holds $\forall k = 0, 1, 2, \cdots$, we formulate their difference as
\begin{subequations}
\label{difference_L}
\begin{align}
& \mathcal{L}({\bf x}^{(\! k+1 \!)}, \{{\bf z}_{i}^{(\! k+1 \!)}\}, \{{\bf u}_{i}^{(\! k+1 \!)}\}) - \mathcal{L}({\bf x}^{(\! k \!)}, \{{\bf z}_{i}^{(\! k \!)}\}, \{{\bf u}_{i}^{(\! k \!)}\}) \nonumber \\
\label{difference_a}
= & \left[ \mathcal{L}( \! {\bf x}^{(\! k + 1 \!)} \! , \! \{ \! {\bf z}_{i}^{(\! k + 1 \!)} \! \} \! , \! \{ \! {\bf u}_{i}^{(\! k + 1 \!)} \!\} \! ) - \mathcal{L}( \! {\bf x}^{(\! k + 1 \!)} \! , \! \{ \! {\bf z}_{i}^{(\! k + 1 \!)} \! \} \! , \! \{ \! {\bf u}_{i}^{(\! k \!)} \! \} \! ) \right] \\
\label{difference_b}
& + \left[ \mathcal{L}({\bf x}^{(\! k+1 \!)} \! , \! \{ \! {\bf z}_{i}^{(\! k+1 \!)} \! \} \! , \! \{ \! {\bf u}_{i}^{(\! k \!)} \! \}) - \mathcal{L}({\bf x}^{(\! k \!)} \! , \! \{ \! {\bf z}_{i}^{(\! k + 1 \!)} \! \} \! , \! \{ \! {\bf u}_{i}^{(\! k \!)} \! \}) \right] \\
\label{difference_c}
& + \left[ \mathcal{L}({\bf x}^{(\! k \!)} \! , \! \{ \! {\bf z}_{i}^{(\! k + 1\!)} \! \} \! , \! \{ \! {\bf u}_{i}^{(\! k \!)} \! \}) - \mathcal{L}({\bf x}^{(\! k \!)} \! , \! \{ \! {\bf z}_{i}^{(\! k \!)} \! \} \! , \! \{ \! {\bf u}_{i}^{(\! k \!)} \! \}) \right].
\end{align}
\end{subequations}

In the sequel, we successively deal with (\ref{difference_a}), (\ref{difference_b}), and (\ref{difference_c}). Firstly, for (\ref{difference_a}), it is calculated as
\begin{align}
\label{long_equa}
& \mathcal{L}( {\bf x}^{(\! k + 1 \!)} , \{ \! {\bf z}_{i}^{(\! k + 1 \!)} \! \} , \{ \! {\bf u}_{i}^{(\! k + 1 \!)} \!\} ) - \mathcal{L}( {\bf x}^{(\! k + 1 \!)} , \{ \! {\bf z}_{i}^{(\! k + 1 \!)} \! \} , \{ \! {\bf u}_{i}^{(\! k \!)} \! \} ) \nonumber \\
\stackrel{\text{(\ref{long_equa}a)}}{=} ~\! & \rho \! \sum_{i = 1}^{m }  \left(  \left\|{\bf z}_{i}^{(\! k+1 \!)} - {\bf x}^{(\! k+1 \!)} + {\bf u}_{i}^{(\! k+1 \!)} \right\|_{2}^{2} - \left\|{\bf z}_{i}^{(\! k+1 \!)} - {\bf x}^{(\! k+1 \!)} + {\bf u}_{i}^{(\! k \!)} \right\|_{2}^{2} - \left\|{\bf u}_{i}^{(\! k+1 \!)}\right\|_{2}^{2} + \left\|{\bf u}_{i}^{(\! k \!)}\right\|_{2}^{2}  \right)  \nonumber  \\
\stackrel{\text{(\ref{long_equa}b)}}{=} ~\! & \rho \sum_{i = 1}^{m} \left( \left\| 2{{\bf u}}_{i}^{(\! k+1 \!)} - {{\bf u}}_{i}^{(\! k \!)} \right\|_{2}^{2}  -  2 \left\|{{\bf u}}_{i}^{(\! k+1 \!)}\right\|_{2}^{2} + \left\|{{\bf u}}_{i}^{(\! k \!)}\right\|_{2}^{2} \right)  \nonumber  \\
= ~~\! & 2 \rho \sum_{i = 1}^{m} \left\| {{\bf u}}_{i}^{(\! k+1 \!)} - {{\bf u}}_{i}^{(\! k \!)} \right\|_{2}^{2} ~
\stackrel{\text{(\ref{long_equa}c)}}{\leq} ~ \rho C \left\| \sum_{i = 1}^{m}\left({\bf u}_{i}^{(\! k+1\!)} - {\bf u}_{i}^{(\! k \!)} \right) \right\|_{2}^{2} ~
\stackrel{\text{(\ref{long_equa}d)}}{\leq} ~ \frac{C}{\rho} \lambda_{\mathrm{max}}^{2}({\bf A}_{0}) \left\| {\bf x}^{(\! k+1 \!)} - {\bf x}^{(\! k \!)} \right\|_{2}^{2},
\end{align}
where in (\ref{long_equa}a) we use the definition of $\mathcal{L}({\bf x}, \{{\bf z}_{i}\}, \{{\bf u}_{i}\})$; in (\ref{long_equa}b) we use (\ref{ADMM_u_k1}); in (\ref{long_equa}c) we use the fact discussed in Remark \ref{upperbound_L2_norm_u} as below; in (\ref{long_equa}d) we employ (\ref{u_and_x}) and the inequality $\|{{\bf H}}{\bf a}\|_{2}^{2} \leq \lambda_{\mathrm{max}}^{2}({\bf H})\|{\bf a}\|_{2}^{2}$ for any Hermitian matrix ${\bf H} \in \mathbb{C}^{n \times n}$ and any vector ${\bf a} \in \mathbb{C}^{n}$.

Next, we move on to (\ref{difference_b}), which is calculated as
\begin{align}
\label{L_x_k1k}
& \mathcal{L}({\bf x}^{(\! k+1 \!)} , \{ \! {\bf z}_{i}^{(\! k+1 \!)} \! \} , \{ \! {\bf u}_{i}^{(\! k \!)} \! \}) - \mathcal{L}({\bf x}^{(\! k \!)} , \{ \! {\bf z}_{i}^{(\! k+1 \!)} \! \} , \{ \! {\bf u}_{i}^{(\! k \!)} \! \}) \nonumber \\
\stackrel{\text{(\ref{L_x_k1k}a)}}{\leq} ~\! & \Re \! \left\{ \left\langle \nabla_{{\bf x}}\mathcal{L}({\bf x}^{(\! k+1 \!)}, \{{\bf z}_{i}^{(\! k+1 \!)}\}, \{{\bf u}_{i}^{(\! k \!)}\}) , {\bf x}^{(\! k+1 \!)} \! - \! {\bf x}^{(\! k \!)} \right\rangle \right\}  - \frac{\gamma_{{\bf x}}}{2}\left\| {\bf x}^{(\! k + 1 \!)} - {\bf x}^{(\! k \!)} \right\|_{2}^{2} \nonumber \\
\stackrel{\text{(\ref{L_x_k1k}b)}}{=} ~\!\! & - \! \left( \lambda_{\mathrm{min}}({\bf A}_{0}) + {m\rho} \right) \left\| {\bf x}^{(\! k + 1 \!)} - {\bf x}^{(\! k \!)} \right\|_{2}^{2},
\end{align}
where in (\ref{L_x_k1k}a) we utilize the fact that function $\mathcal{L}({\bf x} , \{ \! {\bf z}_{i} \! \} , \{ \! {\bf u}_{i} \! \})$ is strongly convex (see Remark \ref{strongconvex_remark} below) with respect to (w.r.t.) ${\bf x}$ with parameter $\gamma_{{\bf x}} > 0$ \cite{Ryu2016}; in (\ref{L_x_k1k}b) we employ the optimality condition, namely, $\nabla_{{\bf x}}\mathcal{L}({\bf x}^{(\! k+1 \!)}, \{{\bf z}_{i}^{(\! k+1 \!)}\}, \{{\bf u}_{i}^{(\! k \!)}\}) = {\bf 0}$, and $\gamma_{{\bf x}} = 2\lambda_{\mathrm{min}}({\bf A}_{0}) + 2m\rho$ (see Remark \ref{strongconvex_remark}).

Lastly, (\ref{difference_c}) is calculated as
\begin{align}
\label{L_z_k1k}
& \mathcal{L}({\bf x}^{(\! k \!)} , \{ \! {\bf z}_{i}^{(\! k +1 \!)} \! \} , \{ \! {\bf u}_{i}^{(\! k \!)} \! \}) - \mathcal{L}({\bf x}^{(\! k \!)} , \{ \! {\bf z}_{i}^{(\! k \!)} \! \} , \{ \! {\bf u}_{i}^{(\! k \!)} \! \}) \nonumber \\
\stackrel{\text{(\ref{L_z_k1k}a)}}{=} ~\! & \rho  \sum_{i = 1}^{m } \!\! \left( \! \left\|{\bf z}_{i}^{(\! k+1 \!)} \! - \! {\bf x}^{(\! k \!)} \! + \! {\bf u}_{i}^{(\! k \!)} \right\|_{2}^{2} \! - \! \left\|{\bf z}_{i}^{(\! k \!)} \! - \! {\bf x}^{(\! k \!)} \! + \! {\bf u}_{i}^{(\! k \!)} \right\|_{2}^{2} \! \right) \stackrel{\text{(\ref{L_z_k1k}b)}}{\leq} ~\!  0,
\end{align}
where we utilize the definition of $\mathcal{L}({\bf x}, \{{\bf z}_{i}\}, \{{\bf u}_{i}\})$ in (\ref{L_z_k1k}a); in (\ref{L_z_k1k}b) we use the fact that ${\bf z}_{i}^{(\!k+1\!)}$ is optimal to (\ref{ADMM_z_k1}).

Therefore, by substituting inequalities (\ref{long_equa}), (\ref{L_x_k1k}), and (\ref{L_z_k1k}) back into (\ref{difference_L}), we have
\begin{align*}
&  \mathcal{L}({\bf x}^{(\! k+1 \!)}, \{{\bf z}_{i}^{(\! k+1 \!)}\}, \{{\bf u}_{i}^{(\! k+1 \!)}\}) - \mathcal{L}({\bf x}^{(\! k \!)}, \{{\bf z}_{i}^{(\! k \!)}\}, \{{\bf u}_{i}^{(\! k \!)}\})  \\
 \leq & \underbrace{ \left( \! \frac{C}{\rho} \lambda_{\mathrm{max}}^{2} \! - \! \lambda_{\mathrm{min}} \! - \! {m\rho} \! \right) \! \left\| {\bf x}^{(\! k+1 \!)} \! - \! {\bf x}^{(\! k \!)} \right\|_{2}^{2} }_{\text{(i)}}.
\end{align*}

We provide the following discussion regarding the term (i) of the above inequality. It is seen that if 
\begin{align*}
\rho < {\left( - \sqrt{\lambda_{\mathrm{min}}^{2} + 4mC\lambda_{\mathrm{max}}^{2} } -\lambda_{\mathrm{min}} \right)} / {(2m)}
\end{align*}
(which should be deleted because of $\rho > 0$) or 
\begin{align*}
\rho > {\left( \sqrt{\lambda_{\mathrm{min}}^{2} + 4mC\lambda_{\mathrm{max}}^{2} } - \lambda_{\mathrm{min}} \right)} / {(2m)},
\end{align*}
the coefficient $\frac{C}{\rho} \lambda_{\mathrm{max}}^{2} \! - \! \lambda_{\mathrm{min}} \! - \! {m\rho} < 0$ and hence $(\text{i}) \leq 0$. Furthermore, as $\sqrt{\lambda_{\mathrm{min}}^{2} \! + \! 4mC\lambda_{\mathrm{max}}^{2} } \! - \! \lambda_{\mathrm{min}} \leq 2\sqrt{mC}\lambda_{\mathrm{max}} + |\lambda_{\mathrm{min}}| - \lambda_{\mathrm{min}} = 2\sqrt{mC}\lambda_{\mathrm{max}} + 2 \max\{ -\lambda_{\mathrm{min}} , 0 \}$, we reach the following conclusion: $(\text{i}) \leq 0$, as long as 
\begin{align}
\label{rho_positive}
\rho > \left( {\sqrt{mC}} \lambda_{\mathrm{max}} + \max\{ -\lambda_{\mathrm{min}} , 0 \} \right)/m.
\end{align}

Therefore, combining  (\ref{rho_positive}) and (\ref{rho_stongconvex}), we conclude that $\mathcal{L}({\bf x}^{(\! k+1 \!)} \! , \! \{{\bf z}_{i}^{(\! k+1 \!)}\} \! , \! \{{\bf u}_{i}^{(\! k+1 \!)}\} \! ) - \mathcal{L}({\bf x}^{(\! k \!)} \! , \! \{{\bf z}_{i}^{(\! k \!)}\} \! , \! \{{\bf u}_{i}^{(\! k \!)}\} \!) \leq (\text{i}) \leq 0$ as long as $\rho$ satisfies (\ref{rho_condition}).

To prove the second part of the theorem, we first define $\Phi \triangleq  - \frac{C}{\rho} \lambda_{\mathrm{max}}^{2} + \lambda_{\mathrm{min}} + {m\rho} $, and denote $\mathcal{L}( {\bf x}^{(\!k+1\!)} \! , \! \{ \! {\bf z}_{i}^{(\!k+1\!)} \! \} \! , \! \{ \! {\bf u}_{i}^{(\!k+1\!)} \!\} )$ and $\mathcal{L}( {\bf x}^{(\!k\!)} \! , \! \{ \! {\bf z}_{i}^{(\!k\!)} \! \} \! , \! \{ \! {\bf u}_{i}^{(\!k\!)} \!\} )$ by $\mathcal{L}^{(\!k+1\!)}$ and $\mathcal{L}^{(\!k\!)}$, respectively. Then, when (\ref{rho_condition}) holds, we have $\Phi > 0$ and
\begin{align*}
\mathcal{L}^{(\!k\!)} - \mathcal{L}^{(\!k+1\!)} \geq \Phi\|{\bf x}^{(\!k+1\!)} - {\bf x}^{(\!k\!)} \|_{2}^{2}, ~ \forall k = 0, 1, 2, \cdots.
\end{align*}
Summing all the above inequalities over all $k \geq 0$, we obtain
\begin{align*}
\mathcal{L}^{(\!0\!)} \geq \Phi \sum_{k=0}^{\infty}\|{\bf x}^{(\!k+1\!)} \!-\! {\bf x}^{(\!k\!)} \|_{2}^{2},
\end{align*} 
which implies \cite{Lu2015} that $\displaystyle \lim_{k \to \infty} \! \|{\bf x}^{(\!k+1\!)} \!-\! {\bf x}^{(\!k\!)} \|_{2} = 0$. 
\end{proof}

\begin{remark}
\label{upperbound_L2_norm_u}
It is easy to see that $\left\|{\bf u}_{i}^{(\!k+1\!)} - {\bf u}_{i}^{(\!k\!)}\right\|_{2}^{2}$, $\forall i = 1, 2, \cdots, m$, are bounded from above. Hence, there exists an upper bound for $\displaystyle \sum_{i = 1}^{m}\left\|{\bf u}_{i}^{(\!k+1\!)} - {\bf u}_{i}^{(\!k\!)}\right\|_{2}^{2}$. We can find a real-valued constant $C$, such that
\begin{align}
\label{inquation_u}
\displaystyle \sum_{i = 1}^{m}\left\|{\bf u}_{i}^{(\!k+1\!)} - {\bf u}_{i}^{(\!k\!)}\right\|_{2}^{2} \leq C \left\| \displaystyle \sum_{i = 1}^{m} \left( {\bf u}_{i}^{(\!k+1\!)} - {\bf u}_{i}^{(\!k\!)} \right) \right\|_{2}^{2}.
\end{align}
In the case where $m = 1$ (i.e., only one constraint in Problem (\ref{general_QCQP_prob})), the above inequality reduces to
\begin{align}
\label{inquation_u_m1}
\left\|{\bf u}_{1}^{(\!k+1\!)} - {\bf u}_{1}^{(\!k\!)}\right\|_{2}^{2} \leq C \left\|  {\bf u}_{1}^{(\!k+1\!)} - {\bf u}_{1}^{(\!k\!)}  \right\|_{2}^{2},
\end{align}
which implies that the constant $C$ can be selected as 
\begin{align}
\label{C_m1}
\left\{
\begin{array}{cc}
C \text{~\! is arbitrary},  & \text{if ~\!} {\bf u}_{1}^{(\!k+1\!)} = {\bf u}_{1}^{(\!k\!)}, \\
C \geq 1, & \text{otherwise}.
\end{array}
\right.
\end{align}

In what follows, we provide an explanation on how to choose $C$ from a statistical perspective. For the sake of notational simplicity, we define ${\bf d}_{i} \triangleq {\bf u}_{i}^{(\!k+1\!)} - {\bf u}_{i}^{(\!k\!)} \in \mathbb{C}^{n}$, $\forall i = 1, 2, \cdots, m$, and ${\bf d}_{i} = [d_{i(1)}, d_{i(2)}, \cdots, d_{i(n)}]^{\mathrm{T}}$ where $d_{i(j)}$ is the $j$-th entry of ${\bf d}_{i}$. We make the following assumptions: 
\begin{enumerate}[label={\textnormal{A\arabic*)}}]
\item ${\bf d}_{i}$'s are independently and identically distributed (i.i.d.) complex random variables; 
\item for any ${\bf d}_{i}$, all its entries are i.i.d.; 
\item $\mathrm{E}\{{\bf d}_{i}\} = {\mu}{\bf 1}$ and $\mathrm{Var}\{{\bf d}_{i}\} = \sigma^{2}{\bf I}_{n}$, $\forall i = 1, 2, \cdots, m$. 
\end{enumerate}
These assumptions indicate that $\mathrm{E}\{d_{i(j)}\} = \mu$ and $\mathrm{Var}\{d_{i(j)}\} = \sigma^{2}$, for all $i = 1, 2, \cdots, m$ and $j = 1, 2, \cdots, n$.

We focus on the expected value of (\ref{inquation_u}), which is given as
\begin{align}
\label{inquation_u_exp}
\displaystyle \mathrm{E}\left\{ \sum_{i = 1}^{m}\left\|{\bf d}_{i}\right\|_{2}^{2} \right\} \leq \mathrm{E} \left\{ C \left\| \displaystyle \sum_{i = 1}^{m} {\bf d}_{i} \right\|_{2}^{2} \right\}.
\end{align}
The left-hand side of (\ref{inquation_u_exp}) is determined as 
\begin{align}
\label{LeftHand}
\displaystyle \mathrm{E}\left\{ \sum_{i = 1}^{m}\left\|{\bf d}_{i}\right\|_{2}^{2} \right\} ~\! = ~\! & \displaystyle \sum_{i = 1}^{m} \mathrm{E}\left\{ \left\|{\bf d}_{i}\right\|_{2}^{2} \right\} \nonumber \\
= ~\! & \displaystyle \sum_{i = 1}^{m} \mathrm{E}\left\{ \sum_{j = 1}^{n} |d_{i(j)}|^{2} \right\} \nonumber \\
= ~\! & \displaystyle \sum_{i = 1}^{m}\sum_{j = 1}^{n} \mathrm{E}\left\{ |d_{i(j)}|^{2} \right\} \nonumber \\
\stackrel{\text{(\ref{LeftHand}{\textnormal{a}})}}{=} & ~\! \displaystyle \sum_{i = 1}^{m}\sum_{j = 1}^{n} \left[ \mathrm{Var}\{ d_{i(j)} \} + \left| \mathrm{E}\{ d_{i(j)} \} \right|^{2} \right] \nonumber \\
= ~\! & \displaystyle \sum_{i = 1}^{m}\sum_{j = 1}^{n} (\sigma^{2} + | \mu |^{2}) \nonumber \\
= ~\! & mn(\sigma^{2} + | \mu |^{2}),
\end{align}
where in {\text{(\ref{LeftHand}{\textnormal{a}})}} we use the relationship between expected value and variance of a complex random variable, see for example Pages 117 and 118 in \cite{Park2018}. Similarly, the right-hand side of (\ref{inquation_u_exp}) is computed as
\begingroup
\allowdisplaybreaks
\begin{align}
\label{RightHand}
\mathrm{E} \left\{ C \left\| \displaystyle \sum_{i = 1}^{m} {\bf d}_{i} \right\|_{2}^{2} \right\} ~\! = ~\! & C \times \mathrm{E}\left\{  \left\| \displaystyle \sum_{i = 1}^{m} {\bf d}_{i} \right\|_{2}^{2} \right\} \nonumber \\
= ~\! & C \times \mathrm{E}\left\{ \displaystyle \sum_{j = 1}^{n} \left| \displaystyle \sum_{i = 1}^{m} d_{i(j)} \right|^{2}  \right\} \nonumber \\
= ~\! & C \times \displaystyle \sum_{j = 1}^{n} \mathrm{E}\left\{ \left| \displaystyle \sum_{i = 1}^{m} d_{i(j)} \right|^{2}  \right\} \nonumber \\
\stackrel{\text{(\ref{RightHand}{\textnormal{a}})}}{=} \! & ~ C \times \displaystyle \sum_{j = 1}^{n} \left[ \mathrm{Var}\left\{ \displaystyle \sum_{i = 1}^{m} d_{i(j)} \right\} + \left| \mathrm{E}\left\{ \displaystyle \sum_{i = 1}^{m} d_{i(j)} \right\} \right|^{2} \right] \nonumber \\
= ~\! & C \times \displaystyle \sum_{j = 1}^{n} \left[ \displaystyle \sum_{i = 1}^{m} \mathrm{Var}\left\{ d_{i(j)} \right\} + \left| \displaystyle \sum_{i = 1}^{m} \mathrm{E}\left\{ d_{i(j)} \right\} \right|^{2} \right] \nonumber \\
= ~\! & C \times \displaystyle \sum_{j = 1}^{n} \left( m\sigma^{2} + m^{2} |\mu|^{2} \right) \nonumber \\
= ~\! & Cmn(\sigma^{2} + m |\mu|^{2}),
\end{align}
\endgroup
where in {\text{(\ref{RightHand}{\textnormal{a}})}} we use the same strategy as in {\text{(\ref{LeftHand}{\textnormal{a}})}}. Substituting (\ref{LeftHand}) and (\ref{RightHand}) into (\ref{inquation_u_exp}) yields
\begin{align}
\label{C_ineq}
\sigma^{2} + |\mu|^{2} \leq C(\sigma^{2} + m |\mu|^{2}),
\end{align}
which implies that
\begin{align}
\label{C_m}
\left\{
\begin{array}{cc}
C \text{~\! is arbitrary},  & \text{if ~\!} \mu = 0 \text{~\! and ~\!} \sigma^{2} = 0, \\
C \geq \frac{\sigma^{2} + |\mu|^{2}}{\sigma^{2} + m |\mu|^{2}}, & \text{otherwise}.
\end{array}
\right.
\end{align}
Note that: {\textnormal{i)}} the case where $\mu = 0$ and $\sigma^{2} = 0$ indicates that ${\bf d}_{i} = {\bf 0}$ (i.e., ${\bf u}_{i}^{(\!k+1\!)} = {\bf u}_{i}^{(\!k\!)}$), $\forall i = 1, 2, \cdots, m$; and {\textnormal{ii)}} by assigning $m = 1$, we have $\frac{\sigma^{2} + |\mu|^{2}}{\sigma^{2} + m |\mu|^{2}} = 1$. Therefore, when $m = 1$, (\ref{C_m}) is consistent with (\ref{C_m1}).
\end{remark}

\begin{remark}
\label{strongconvex_remark}
The statement that $\mathcal{L}( {\bf x}, \! \{ \! {\bf z}_{i} \! \}, \! \{ \! {\bf u}_{i} \! \} )$ is strongly convex w.r.t. ${\bf x}$ with parameter $\gamma_{{\bf x}}$, is equivalent to the statement that $\mathcal{L}({\bf x} , \{ \! {\bf z}_{i} \! \} , \{ \! {\bf u}_{i} \! \}) - \frac{\gamma_{{\bf x}}}{2}\|{\bf x}\|_{2}^{2}$ is convex w.r.t. ${\bf x}$, see \cite{Zhou2018}. Note that $\mathcal{L}({\bf x} , \{ \! {\bf z}_{i} \! \} , \{ \! {\bf u}_{i} \! \}) - \frac{\gamma_{{\bf x}}}{2}\|{\bf x}\|_{2}^{2} = {\bf x}^{\mathrm{H}} [ {\bf A}_{0} + (m\rho - \frac{\gamma_{{\bf x}}}{2}){\bf I}_{n} ] {\bf x} - 2 \Re \! \left\{ \left( {\bf b}_{0} \!+\! \rho\sum_{i = 1}^{m}({\bf z}_{i} + {\bf u}_{i}) \right)^{\mathrm{H}}{\bf x} \right\} + C'$, where $C'$ is a constant independent of ${\bf x}$. $\mathcal{L}({\bf x} , \{ \! {\bf z}_{i} \! \} , \{ \! {\bf u}_{i} \! \}) - \frac{\gamma_{{\bf x}}}{2}\|{\bf x}\|_{2}^{2}$ is convex if and only if ${\bf A}_{0} + (m\rho - \frac{\gamma_{{\bf x}}}{2}){\bf I}_{n} \succeq 0$, which indicates that parameter $\gamma_{{\bf x}} = 2\lambda_{\mathrm{min}}({\bf A}_{0}) + 2m\rho$. This equality together with ${\gamma_{{\bf x}}} > 0$ yields
$\rho > - \lambda_{\mathrm{min}}({\bf A}_{0})/m$. Since $\rho > 0$, we have 
\begin{align}
\label{rho_stongconvex}
\rho > \max\left\{ - \lambda_{\mathrm{min}}({\bf A}_{0})/m , 0\right\}.
\end{align}
\end{remark}

\begin{remark}
\label{boundedness_remark}
When ${\bf A}_{0} \succ 0$, our simulation results show that the augmented Lagrangian function $\mathcal{L}({\bf x}, \{{\bf z}_{i}\}, \{{\bf u}_{i}\})$ is bounded from below during the consensus-ADMM iteration, i.e., (\ref{ADMM_k1_prob}). This observation together with the result in Theorem indicates the convergence of the consensus-ADMM. However, when ${\bf A}_{0}$ is indefinite, our simulation results show that $\mathcal{L}({\bf x}, \{{\bf z}_{i}\}, \{{\bf u}_{i}\})$ is unbounded.
\end{remark}

\section{Simulation Results}
\label{simulation_section}
In this section, we provide numerical simulations to confirm our theoretical analyses in Section \ref{convergenceproperty_section}. After fixing the problem dimensions $n$ and $m$, we first generate ${\bf x}_{\text{feas}} \sim \mathcal{CN}({\bf 0}, {\bf I}_{n})$. A Hermitian indefinite matrix ${\bf A}_{i}$ ($i = 0, 1, \cdots, m$) is generated by first randomly drawing an $n \times n$ matrix from $\mathcal{CN}(0, 1)$, and then taking the average of its Hermitian and itself. In the situation where we need ${\bf A}_{0} \succ 0$, we add $\epsilon {\bf I}_{n}$ to its indefinite counterpart, with $\epsilon$ larger than the absolute of the smallest eigenvalue of the indefinite matrix. We generate ${\bf b}_{i} \sim \mathcal{CN}({\bf 0}, {\bf I}_{n})$, $\forall i = 0, 1, \cdots, m$, and $c_{i} = {\bf x}_{\text{feas}}^{\mathrm{H}}{\bf A}_{i}{\bf x}_{\text{feas}} - 2\Re\{{\bf b}_{i}^{\mathrm{H}}{\bf x}_{\text{feas}}\} + |{v}_{i}|$, $\forall i = 1, 2, \cdots, m$, where $v_{i}$ is randomly drawn from $\mathcal{N}(0, 1)$. The constructed constraint set is thus guaranteed to be non-empty, because at least ${\bf x}_{\text{feas}}$ is feasible. The above-described problem setting is similar to those considered in \cite{Huang2016} and \cite{Mehanna2015}. We initiate the consensus-ADMM algorithm with ${\bf x}^{(\!0\!)} = {\bf x}_{\text{feas}}$ and ${\bf u}_{i}^{(\!0\!)} = {\bf 0}$, $\forall i = 1, 2, \cdots, m$.

In the first example, we set $n = 10$ and $m = 5$. All ${\bf A}_{i}$'s are indefinite. The augmented Lagrangian function value versus iteration index is plotted in Fig. \ref{AugLagvsIterIndex_IndefiniteA0}. It is seen that the function value is monotonically non-increasing, but unbounded.

In all the remaining examples, ${\bf A}_{i}$ ($i = 1, 2, \cdots, m$) are indefinite and ${\bf A}_{0} \succ 0$. In the second example, $n = 10$ and $m = 5$. The augmented Lagrangian function value versus iteration index is shown in Fig. \ref{AugLagvsIterIndex}. It is observed that the function value is not monotonic when $\rho = 2$ and $\rho = 5$; while it is monotonically non-increasing and bounded from below when $\rho \geq 10$. In Fig. \ref{AugLagvsIterIndex_MC}, we plot the function value versus $k$ in $20$ independent runs, with $n = 10$, $m = 5$, and $\rho = 10$. The function value versus $k$ is depicted in Fig. \ref{AugLagvsIterIndex_NumbConstraint} for different number of constraints, with $n = 10$ and $\rho = 20$. All the curves in Figs. \ref{AugLagvsIterIndex_MC} and \ref{AugLagvsIterIndex_NumbConstraint} indicate the convergence of the augmented Lagrangian function value sequence $\{ \mathcal{L}^{(\!k\!)} \}$.

In the last example, we test the distance of the variables, by using the $\ell_{2}$-norm as performance metric. We set $n = 10$, $m = 5$, and $\rho = 10$. The results are shown in Fig. \ref{DistancevsIterIndex}. It is seen that, both $\left\|{\bf x}^{(\! k + 1 \!)} - {\bf x}^{( \! k \! )} \right\|_{2}$ and $\displaystyle \sum_{i = 1}^{m} \left\| {\bf z}_{i}^{(\! k \!)} - {\bf x}^{(\! k \!)} \right\|_{2}$ approach 0 when $k$ is large enough.

\section{Conclusion}
\label{conclusion_section}
We studied the convergence property of the consensus-alternating direction method of multipliers (ADMM) for general quadratically constrained quadratic programs (QCQPs). We have proved that the augmented Lagrangian function is monotonically non-increasing if the augmented Lagrangian parameter is larger than a certain value. Simulation results showed that the augmented Lagrangian function is bounded from below when the matrix in the quadratic term of the objective function is positive definite. In such cases, the consensus-ADMM was shown to generate a convergent Lagrangian function value sequence. However, the augmented Lagrangian function is unbounded when the matrix in the quadratic term of the objective function is indefinite. We further proved that the distance of the variables between two successive iterations converges to $0$. Numerical simulations with definite as well as indefinite matrices in the QCQP were conducted. The simulation results verified the theoretical development. 

\balance
\biboptions{numbers,sort&compress}
\bibliography{refs}

\newpage

\begin{figure}[t]
	\centerline{\includegraphics[width=0.85\textwidth]{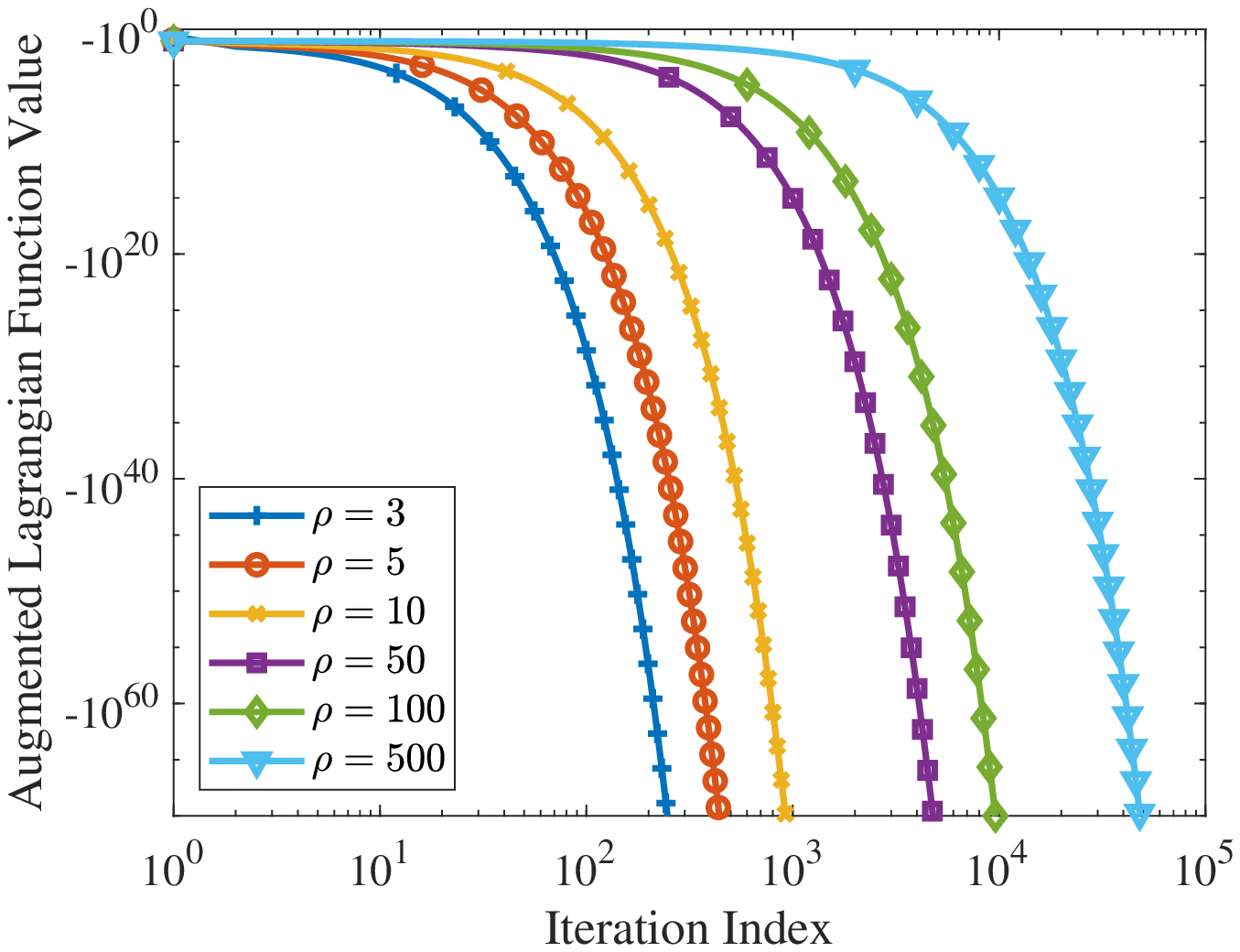}}
	\caption{Function value $\mathcal{L}^{(\!k\!)}$ versus iteration index $k$, with $n = 10$, $m = 5$, and indefinite ${\bf A}_{0}$.}
	\label{AugLagvsIterIndex_IndefiniteA0}
\end{figure}

\begin{figure}[t]
	\centerline{\includegraphics[width=0.85\textwidth]{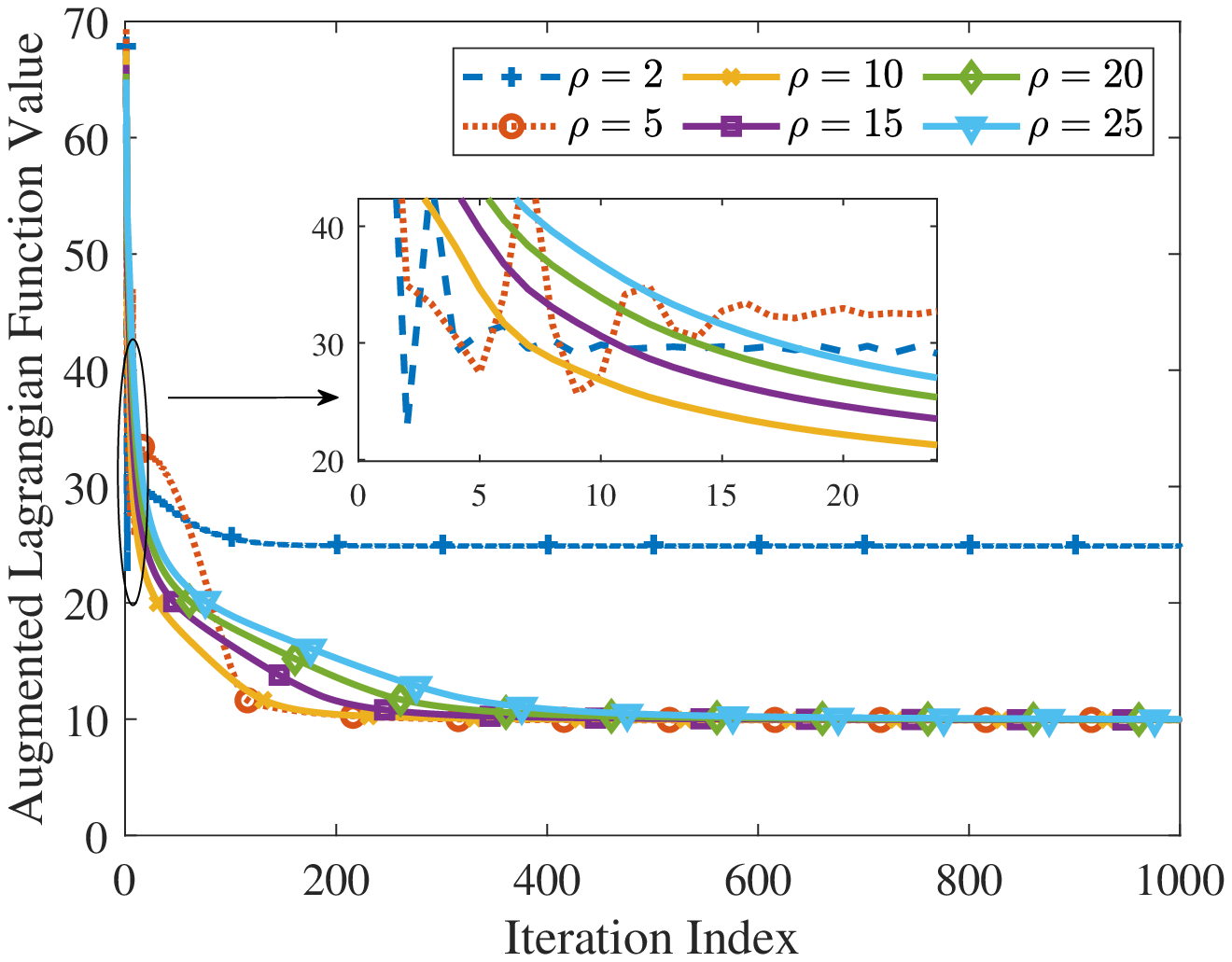}}
	\caption{Function value $\mathcal{L}^{(\!k\!)}$ versus iteration index $k$, with $n = 10$, $m = 5$, and ${\bf A}_{0} \succ 0$.}
	\label{AugLagvsIterIndex}
\end{figure}

\begin{figure}[t]
	\centerline{\includegraphics[width=0.85\textwidth]{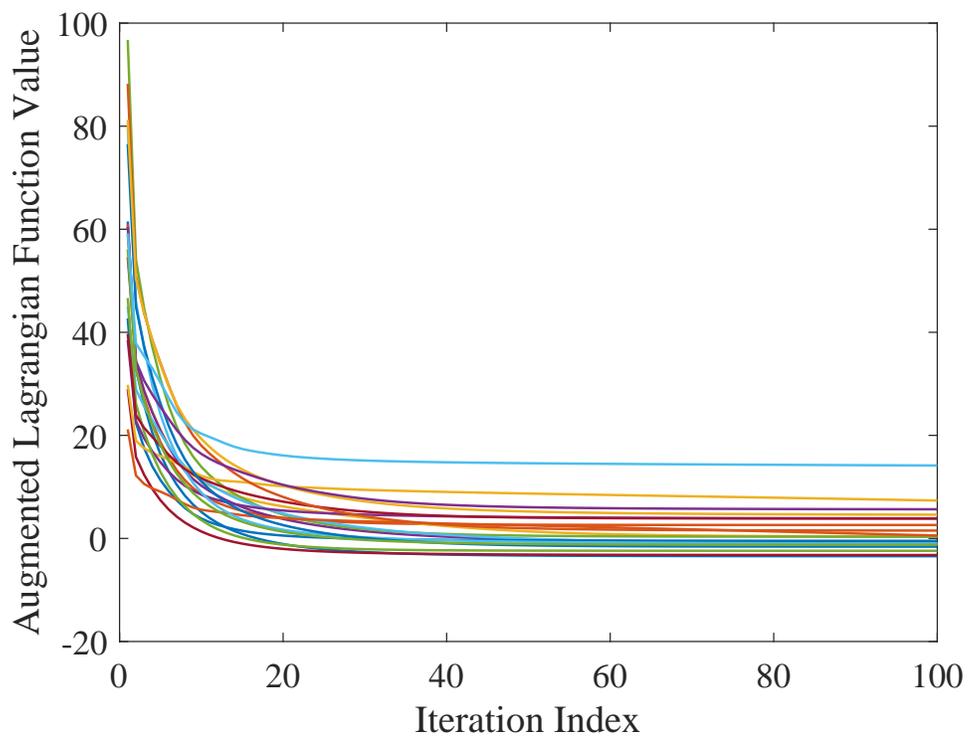}}
	\caption{Function value $\mathcal{L}^{(\!k\!)}$ versus iteration index $k$ in $20$ runs, with $n = 10$, $m = 5$, $\rho = 10$, and ${\bf A}_{0} \succ 0$.}
	\label{AugLagvsIterIndex_MC}
\end{figure}

\begin{figure}[t]
	\centerline{\includegraphics[width=0.85\textwidth]{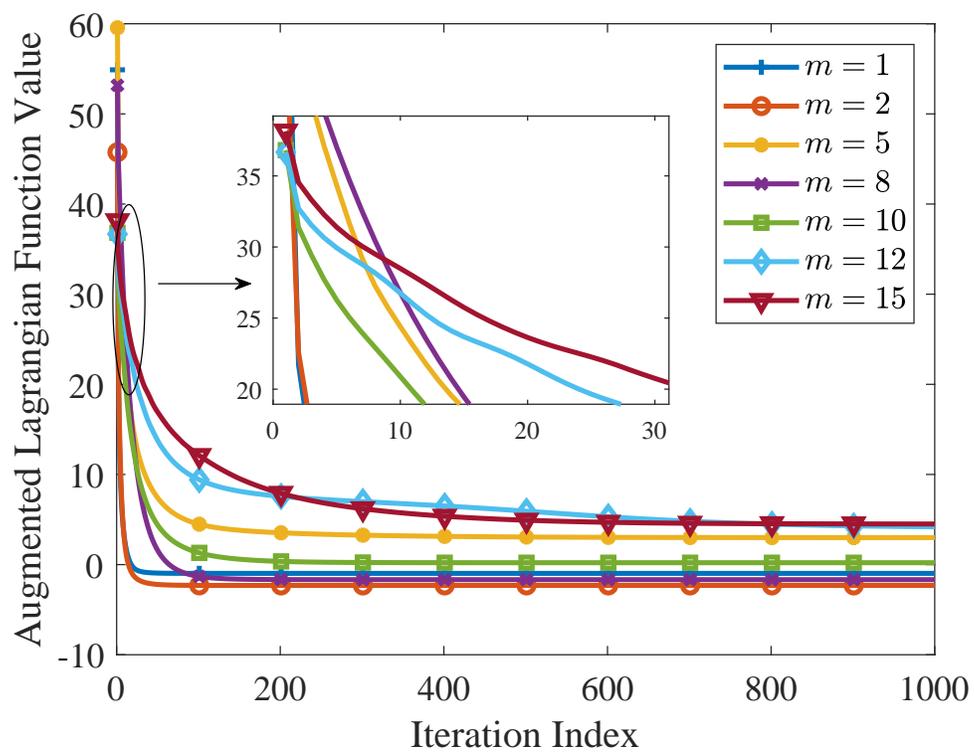}}
	\caption{Function value $\mathcal{L}^{(\!k\!)}$ versus iteration index $k$ for different number of constraints, with $n = 10$, $\rho = 20$, and ${\bf A}_{0} \succ 0$.}
	\label{AugLagvsIterIndex_NumbConstraint}
\end{figure}

\begin{figure}[t]
	\centerline{\includegraphics[width=0.85\textwidth]{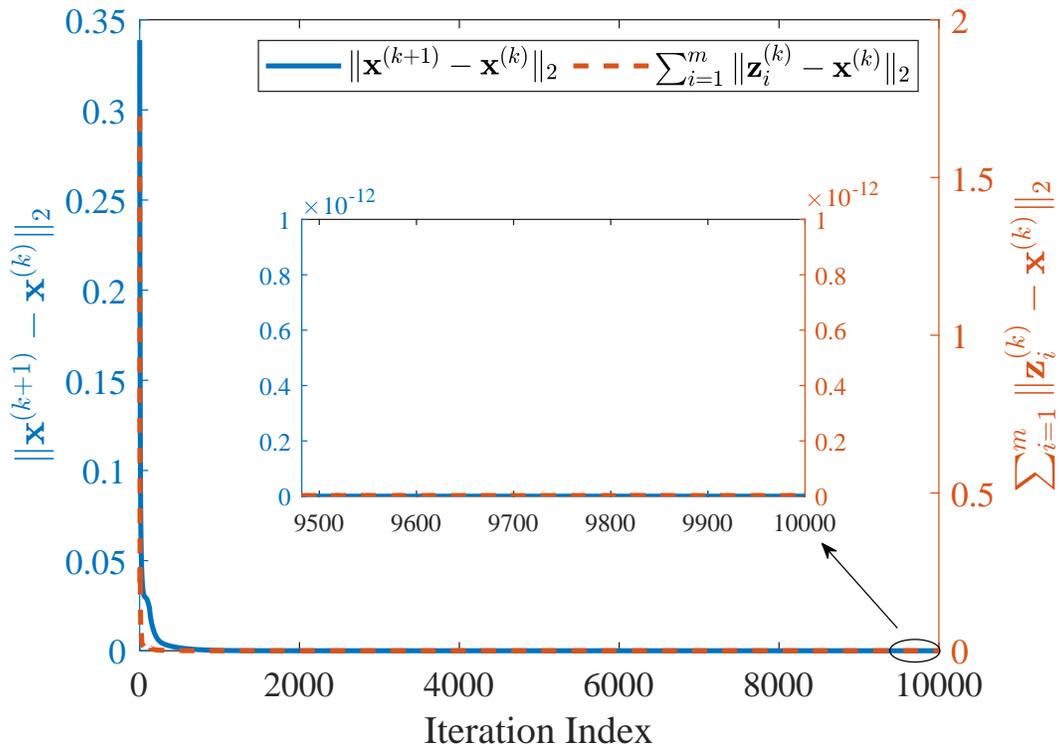}}
	\caption{Distance of variables versus iteration index $k$, with $n = 10$, $m = 5$, $\rho = 10$, and ${\bf A}_{0} \succ 0$.}
	\label{DistancevsIterIndex}
\end{figure}

\end{document}